\documentclass[letter]{article} 
\usepackage[
  margin=3cm,
  includefoot,
  footskip=30pt,
]{geometry}
\usepackage{url}  
\usepackage{graphicx}  

\usepackage{booktabs}       
\usepackage{amsfonts}       
\usepackage{nicefrac}       
\usepackage{microtype}      
\usepackage{amsthm}
\usepackage{algorithm}
\usepackage{bbold}
\usepackage{amsmath}
\usepackage{algorithmic}
\usepackage{verbatim}
\usepackage{subcaption}
\usepackage{cite}
\usepackage{graphicx}
\usepackage{multirow}
\usepackage{amssymb}
\usepackage{color}
\usepackage{caption}
\usepackage{enumitem}
\usepackage{hyphenat}
\usepackage{array}
\usepackage{booktabs}
\usepackage{wrapfig}
\usepackage{relsize}
\usepackage{tikz}

\usepackage{afterpage}
\usetikzlibrary{arrows}
\title{The Geometric Block Model}

\author{
Sainyam Galhotra \quad Arya Mazumdar \quad Soumyabrata Pal \quad Barna Saha \\
  College of Information and Computer Science\\
  University of Massachusetts Amherst\\
  Amherst, MA 01003\\
  \texttt{\{sainyam,arya,spal,barna\}@cs.umass.edu} 
  \thanks{This research is supported by NSF  awards  CCF 1642550, CCF 1464310 and CAREER awards  1642658 and 1652303. A shorter version of this work has appeared in the proceedings of the 32nd AAAI Conference on Artificial Intelligence. The AAAI proceeding version of the paper contained some incorrect results that are corrected in this version.}
}

\DeclareGraphicsExtensions{.eps}

\newcommand\nc\newcommand
\nc\bfa{{\boldsymbol a}}\nc\bfA{{\boldsymbol A}}\nc\cA{{\mathcal A}}
\nc\bfb{{\boldsymbol b}}\nc\bfB{{\boldsymbol B}}\nc\cB{{\mathcal B}}
\nc\bfc{{\boldsymbol c}}\nc\bfC{{\boldsymbol C}}\nc\cC{{\mathcal C}}
\nc\sC{{\mathscr C}}
\nc\bfd{{\boldsymbol d}}\nc\bfD{{\boldsymbol D}}\nc\cD{{\mathcal D}}
\nc\bfe{{\boldsymbol e}}\nc\bfE{{\boldsymbol E}}\nc\cE{{\mathcal E}}
\nc\bff{{\boldsymbol f}}\nc\bfF{{\boldsymbol F}}\nc\cF{{\mathcal F}}
\nc\bfg{{\boldsymbol g}}\nc\bfG{{\boldsymbol G}}\nc\cG{{\mathcal G}}
\nc\bfh{{\boldsymbol h}}\nc\bfH{{\boldsymbol H}}\nc\cH{{\mathcal H}}
\nc\bfi{{\boldsymbol i}}\nc\bfI{{\boldsymbol I}}\nc\cI{{\mathcal I}}
\nc\bfj{{\boldsymbol j}}\nc\bfJ{{\boldsymbol J}}\nc\cJ{{\mathcal J}}
\nc\bfk{{\boldsymbol k}}\nc\bfK{{\boldsymbol K}}\nc\cK{{\mathcal K}}
\nc\bfl{{\boldsymbol l}}\nc\bfL{{\boldsymbol L}}\nc\cL{{\mathcal L}}
\nc\bfm{{\boldsymbol m}}\nc\bfM{{\boldsymbol M}}\nc\sM{{\mathscr M}}\nc\cM{{\mathcal M}}
\nc\bfn{{\boldsymbol n}}\nc\bfN{{\boldsymbol N}}\nc\cN{{\mathcal N}}
\nc\bfo{{\boldsymbol o}}\nc\bfO{{\boldsymbol O}}\nc\cO{{\mathcal O}}
\nc\bfp{{\boldsymbol p}}\nc\bfP{{\boldsymbol P}}\nc\cP{{\mathcal P}}
\nc\bfq{{\boldsymbol q}}\nc\bfQ{{\boldsymbol Q}}\nc\cQ{{\mathcal Q}}
\nc\bfr{{\boldsymbol r}}\nc\bfR{{\boldsymbol R}}\nc\cR{{\mathcal R}}
\nc\bfs{{\boldsymbol s}}\nc\bfS{{\boldsymbol S}}\nc\cS{{\mathcal S}}
\nc\bft{{\boldsymbol t}}\nc\bfT{{\boldsymbol T}}\nc\cT{{\mathcal T}}
\nc\bfu{{\boldsymbol u}}\nc\bfU{{\boldsymbol U}}\nc\cU{{\mathcal U}}
\nc\bfv{{\boldsymbol v}}\nc\bfV{{\boldsymbol V}}\nc\cV{{\mathcal V}}
\nc\bfw{{\boldsymbol w}}\nc\bfW{{\boldsymbol W}}\nc\cW{{\mathcal W}}
\nc\bfx{{\boldsymbol x}}\nc\bfX{{\boldsymbol X}}\nc\cX{{\mathcal X}}
\nc\bfy{{\boldsymbol y}}\nc\bfY{{\boldsymbol Y}}\nc\cY{{\mathcal Y}}
\nc\bfz{{\boldsymbol z}}\nc\bfZ{{\boldsymbol Z}}\nc\cZ{{\mathcal Z}}

\nc\diff{{\mathrm d}}
\nc\e{{\mathrm e}}
\nc\calC{{\mathcal C}}

\newcommand{\remove}[1]{}
\newcommand{\correct}[1]{\textcolor{black}{#1}}
\newcommand{\latest}[1]{\textcolor{red}{#1}}

\newcommand{\avg}{{\mathbb E}}
\newcommand{\expect}{{\mathbb E}}
\newcommand{\dist}{d_{L}}

\newtheorem*{lemma*}{Lemma}

\newtheorem{theorem}{Theorem}
\newtheorem{lemma}{Lemma}

\theoremstyle{definition}
\newtheorem*{definition}{Definition}

\newtheorem{remark}{Remark}

\def\DEBUG{true}

\ifdefined\DEBUG

  \def\rem#1{{\marginpar{\raggedright\scriptsize #1}}}
  \newcommand{\barnr}[1]{\rem{\textcolor{red}{$\bullet$ #1}}}
  \newcommand{\aryar}[1]{\rem{\textcolor{green}{$\bullet$ #1}}}
\else

  \newcommand{\barnr}[1]{}
  \newcommand{\aryar}[1]{}

\fi

\newcommand\reals{{\mathbb R}}

\allowdisplaybreaks

\begin{document}

\maketitle

\begin{abstract}
To capture the inherent geometric features of many community detection problems, we propose to use a new random graph model of communities that we call a {\em Geometric Block Model}. The geometric block model generalizes the random geometric graphs in the same way that the well-studied stochastic block model generalizes the Erd\"{o}s-Renyi random graphs. It is also a natural extension of random community models inspired by the recent theoretical and practical advancement in community detection. While being a topic of fundamental theoretical interest, our main contribution is to show that many practical community structures are better explained by the geometric block model. We also show that a simple triangle-counting algorithm to detect communities in the geometric block model is near-optimal. Indeed, even in the regime where the average degree of the graph grows only logarithmically with the number of vertices (sparse-graph), we show that this algorithm performs extremely well, both theoretically and practically. In contrast, the triangle-counting algorithm is far from being optimum for the stochastic block model. We simulate our results on both real and synthetic datasets to show superior performance of both the new model as well as our algorithm. 
\end{abstract}
\section{Introduction}
The {\em planted-partition} model or the {\em stochastic block model} (SBM) is a random graph model for community detection that generalizes the well-known 
Erd\"{o}s-Renyi graphs \cite{holland1983stochastic,dyer1989solution,decelle2011asymptotic,abbe2015community,abh:16,DBLP:conf/colt/HajekWX15,chin2015stochastic,mossel2015consistency}. Consider a graph $G(V,E)$, where $V = V_1 \sqcup V_2 \sqcup \dots \sqcup V_k$ is a disjoint union of $k$ clusters denoted by $V_1, \dots, V_k.$ The edges of the graph are drawn randomly: there is an edge between $u \in V_i$ and $v \in V_j$ with probability $q_{i,j}, 1\le i,j \le k.$
Given the adjacency matrix of such a graph, the task is to find exactly (or approximately) the partition $V_1 \sqcup V_2 \sqcup \dots \sqcup V_k$ of $V$.

This model has been incredibly popular both in theoretical and practical domains of community detection, and the aforementioned references are just a small sample. Recent theoretical works  focus on characterizing sharp threshold of recovering the partition in the SBM. For example, when there are only two communities of exactly equal sizes, and the inter-cluster edge probability is $\frac{b\log n}{n}$ and intra-cluster edge probability is $\frac{a\log n}{n}$, it is known that perfect recovery is possible if and only if $\sqrt{a} - \sqrt{b} > \sqrt{2}$ 
\cite{abh:16,mossel2015consistency}. The regime of the probabilities being $\Theta\Big(\frac{\log n}{n}\Big)$ has been put forward as one of most interesting ones, because  in an Erd\"{o}s-Renyi random graph, this is the threshold for graph connectivity \cite{bollobas1998random}. This result has been subsequently generalized for $k$ communities \cite{abbe2015community,abbe2015recovering,hajek2016achieving} (for constant $k$ or when $k=o(\log{n})$), and under the assumption that the communities are generated according to a probabilistic generative model (there is a prior probability $p_i$ of an element being in the $i$th community) \cite{abbe2015community}. Note that, the results are not only of theoretical interest, many real-world networks exhibit a ``sparsely connected'' community feature \cite{leskovec2008statistical}, and any efficient recovery algorithm for SBM has many potential applications.  

One aspect that the SBM does not account for is a ``transitivity rule'' (`friends having common friends') inherent to many social and other community structures. To be precise, consider any three vertices $x, y$ and $z$. If $x$ and $y$ are connected by an edge (or they are in the same community), and $y$ and $z$ are connected by an edge (or they are in the same community), then it is more likely than not that $x$ and $z$ are connected by an edge. This phenomenon can be seen in many network structures - predominantly in social networks, blog-networks and advertising. SBM, primarily a generalization of Erd\"{o}s-Renyi random graph, does not take into account this characteristic, and in particular, probability of an edge between $x$ and $z$ there is independence of the fact that there exist edges between $x$ and $y$ and $y$ and $z$. However, one needs to be careful such that by allowing such ``transitivity'', the simplicity and elegance of the SBM is not lost.

Inspired by the above question, we propose a random graph community detection model analogous to the stochastic block model, that we call the {\em geometric block model} (GBM). The GBM depends on the basic definition of the {\em random geometric graph} that has found a lot of practical use in wireless networking because of its inclusion of the notion of proximity between nodes \cite{penrose2003random}.

\noindent{\bf Definition.}
A random geometric graph (RGG) on $n$ vertices has parameters $n$, an integer $t> 1$ and  a  real number $\beta \in [-1,1]$. It is defined by assigning a vector $Z_i \in \reals^t$ to vertex $i, 1 \le i, n,$ where $Z_i, 1\le i \le n$ are independent and identical random vectors uniformly distributed in the Euclidean sphere $\cS^{t-1} \equiv \{x\in \reals^t: \|x\|_{\ell_2} =1\}$. There will be an edge between vertices $i$ and $j$ if and only if $\langle Z_i,  Z_j\rangle \ge \beta$. 

Note that, the definition can be further generalized by considering $Z_i$s to have a sample space other than $\cS^{t-1}$, and by using a different notion of distance than inner product (i.e., the Euclidean distance). We simply stated one of the many equivalent definitions \cite{bubecktriangle}.

Random geometric graphs are often proposed as an alternative to Erd\"{o}s-Renyi random graphs. They are quite well studied theoretically (though not nearly as much as the Erd\"{o}s-Renyi graphs), and very precise results exist regarding their connectivity, clique numbers and other structural properties \cite{guptakumarrgg,lognradius,rggmore,rggmore2,rggmore3}. For a survey of early results on geometric graphs and the analogy to results in Erd\"{o}s-Renyi graphs, we refer the reader to \cite{penrose2003random}.
A very interesting question of distinguishing an Erd\"{o}s-Renyi graph from a geometric random graph has also recently been studied  \cite{bubecktriangle}. This  will provide a way to test between the models which better fits a scenario, a potentially great practical use.

As mentioned earlier, the ``transitivity'' feature led to random geometric graphs being used extensively to model wireless networks (for example, see \cite{haenggi2009stochastic,bettstetter2002minimum}). Surprisingly, however, to the best of our knowledge, random geometric graphs are never used to model community detection problems. In this paper we take the first step towards this direction. 
Our main contributions can be classified as follows.
\begin{itemize}[noitemsep,leftmargin=5pt]
\item We define a random generative model to study canonical problems of community detection, called the {\em geometric block model} (GBM). This model takes into account a measure of proximity between nodes and this proximity measure characterizes the likelihood of two nodes being connected when they are in same or different communities. The geometric block model inherits the connectivity properties of the random geometric graphs, in particular the likelihood of ``transitivity'' in triplet of nodes (or more).

\item We experimentally validate the  GBM on various real-world datasets.  We show that many practical community structures exhibit properties of the GBM. We also compare these features with the corresponding notions in SBM to show how GBM better models data in many practical situations.

\item We propose a simple motif-based efficient algorithm for community detection on the GBM. We rigorously show that this algorithm is optimal up to a constant fraction (to be properly defined later) even in the regime of sparse graphs (average degree $\sim \log n$).

\item The motif-counting algorithms are extensively tested on both synthetic and real-world datasets. They exhibit very good performance in three real datasets, compared to the spectral-clustering algorithm (see Section~\ref{sec:exp}). Since simple motif-counting is known to be far from optimum in stochastic block model (see Section~\ref{sec:theory}), these experiments give further validation to GBM as a real-world model. 
\end{itemize}
 
Given  any simple random graph model, it is possible to generalize it to a random block model of communities much in line with the SBM. We however stress that the geometric block model is perhaps the simplest possible model of real-world communities  that also captures the transitive/geometric features of communities. Moreover, the GBM  explains behaviors of many real world networks as we will exemplify subsequently. 


\begin{table*}[htbp]
\begin{tabular}{|p{1.5cm}|p{1.5cm}|p{1.5cm}|p{1.5cm}|} 
 \hline
 Area 1 & Area 2  & same  & different \\ 
 \hline
MOD & AI  & 10& 2 \\
\hline
ARCH & MOD & 6& 1\\
\hline
ROB & ARCH  &3& 0\\
\hline
MOD & ROB & 4 &0\\
\hline
ML & MOD &7&1\\
\hline
\end{tabular}
\quad
\begin{tabular}{|p{1.5cm}|p{1.5cm}|p{1.5cm}|} 
 \hline
 Area  &  same & different   \\ 
 \hline
MOD & 19&35 \\
\hline
ARCH & 13 & 15\\
\hline
ROB & 24 &16\\
\hline
AI & 39& 32\\
\hline
ML &14 &42\\
\hline
\end{tabular}
\caption{On the left we count the number of inter-cluster edges when authors shared same affiliation and different affiliations. On the right, we count the same for intra-cluster edges.\label{table:collab}}
\end{table*}

\section{The Geometric Block Model and its Validation}
Let $V\equiv  V_1 \sqcup V_2 \sqcup \dots \sqcup V_k$ be the set of vertices that is a disjoint union of $k$ clusters, denoted by $V_1, \dots , V_k$.
Given an integer $t\geq 2$, for each vertex $u \in V$, define a random vector $Z_u \in \reals^t$ that is uniformly distributed in $\cS^{t-1} \subset \reals^t,$ the $t-1$-dimensional sphere. 

\begin{definition}[Geometric Block Model ($V, t, \beta_{i,j}, 1\le i <j \le k$)]
Given $V, t$ and a set of real numbers $\beta_{i,j} \in [-1,1], 1\le i \le j \le k$, the geometric block model is a random graph with vertices $V$ and an edge exists between $v \in V_i$ and $u \in V_j$ if and only if $\langle Z_u, Z_v\rangle \ge \beta_{i,j}$.
\end{definition}

\noindent{\bf The case of $t=2$:}
In this paper we particularly analyze our algorithm for $t=2$. In this special case, the above definition is equivalent to choosing  random variable $\theta_u$ uniformly distributed in  $[0,2\pi]$, for all $u \in V$. Then there will be an edge between two vertices $u\in V_i,v\in V_j$ if and only if 
$\cos \theta_u \cos \theta_v + \sin \theta_u \sin \theta_v = \cos(\theta_u -\theta_v) \ge \beta_{i,j}$ or $\min\{|\theta_u -\theta_v|, 2\pi -|\theta_u - \theta_v|\} \le \arccos \beta_{i,j}$. This in turn, is equivalent to choosing a random variable $X_u$ uniformly distributed in $[0,1]$ for all $u \in V$, and there exists an edge between   two vertices $u\in V_i,v\in V_j$ if and only if 
$$
d_L(X_u,X_v) \equiv \min\{|X_u - X_v|, 1- |X_u- X_v|\} \le r_{i,j},
$$
where $r_{i,j} \in [0,\frac12], 0 \le i,j \le k$, are a set of real numbers.
 
For the rest of this paper, we concentrate on the case when $r_{i,i} = r_s$ for all $i \in \{1, \dots, k\}$, which we call the ``intra-cluster distance'' and $r_{i,j} = r_d$ for all $i,  j\in  \{1, \dots, k\}, i \ne j$, which we call the ``inter-cluster distance,'' mainly for the clarity of exposition. To allow for edge density to be higher inside the clusters than across the clusters,  assume $r_s \geq r_d$.

%


The main problem that we seek to address is following. Given the adjacency matrix of a geometric block model with $k$ clusters, and $t, r_d, r_s$, $r_s \geq r_d$, find the partition $V_1, V_2, \dots, V_k$.



 
 We next give two examples of real datasets that motivate the GBM. In particular, we experiment with two different types of real world datasets in order to verify our hypothesis about geometric block model and the role of distance in the formation of edges. The first one is a dataset with academic collaboration, and the second one is a product purchase metadata from Amazon.

 \subsection{Motivation of GBM: Academic Collaboration}
 \label{subsec:m1}
We consider the collaboration network of academicians in Computer Science in 2016 (data obtained from \texttt{csrankings.org}). According to area of expertise of the authors, we 
consider five different communities: Data Management (MOD), Machine Learning and Data Mining (ML), Artificial Intelligence (AI), Robotics (ROB), Architecture (ARCH). If two authors share the same affiliation, or shared affiliation in the past, we assume that they are geographically close.
We would like to hypothesize that, two authors in the same communities might collaborate even when they are geographically far. However, two authors in different communities are more likely to collaborate only if they share the same affiliation (or are geographically close).
Table \ref{table:collab} describes the number of edges across  the communities. 
 It is evident that the authors from same community are likely to collaborate irrespective of the affiliations and the authors of different communities collaborate much frequently when they share affiliations or are close geographically. This clearly indicates that the inter cluster edges are likely to form if the distance between the nodes is quite small, motivating the fact $r_d < r_s$ in the GBM.  


\subsection{Motivation of GBM: Amazon Metadata}
\label{subsec:m2}
The next dataset that we use in our experiments is the Amazon product metadata on SNAP (\url{https://snap.stanford.edu/data/amazon-meta.html}), that has 548552 products and each product is one of the following types
\{Books, Music CD's, DVD's, Videos\}. Moreover, each product has a list of attributes, for example, a book may have attributes like $\langle$``General'', ``Sermon'', ``Preaching''$\rangle$. We consider the co-purchase network over these products. We make two observations here: (1) edges get formed (that is items are co-purchased) more frequently if they are similar, where we measure similarity by the number of common attributes between products, and (2) two products that share an edge have more common neighbors (no of items that are bought along with both those products) than two products with no edge in between. 

Figures~\ref{fig:amazon1} and \ref{fig:amazon2} show respectively average similarity of products that were bought together, and not bought together. From the distribution, it is quite evident that edges in a co-purchase network gets formed according to distance, a salient feature of random geometric graphs, and the GBM. 

\begin{figure*}[htbp]
  \centering
  \begin{minipage}[b]{0.4\textwidth}
    \includegraphics[width=0.8\textwidth]{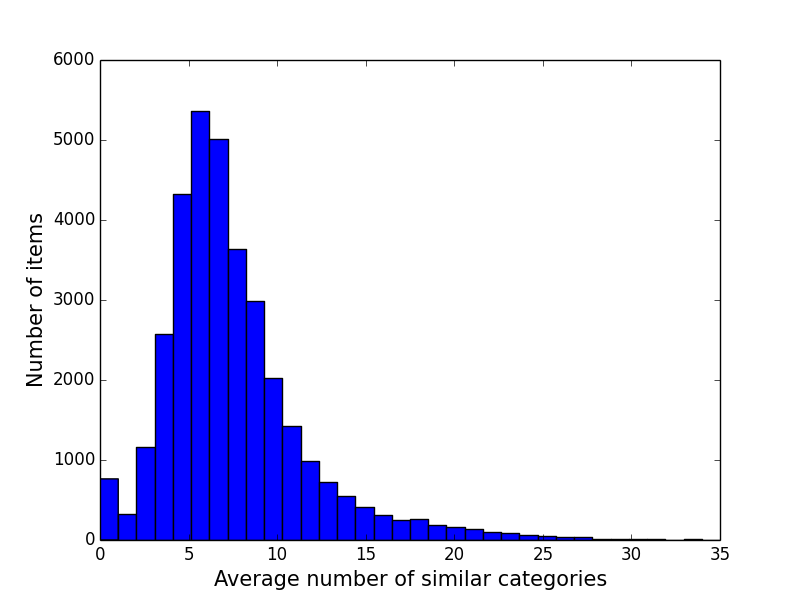}
    \caption{Histogram: similarity of products bought together (mean $\approx 6$)}
    \label{fig:amazon1}
  \end{minipage}
  \quad
  \begin{minipage}[b]{0.4\textwidth}
    \includegraphics[width=0.8\textwidth]{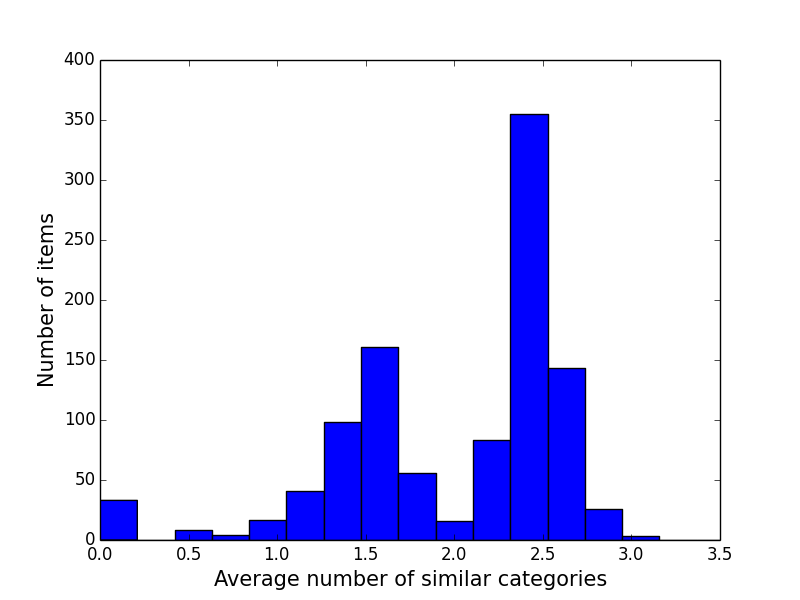}
    \caption{Histogram: similarity of products not bought together (mean$\approx 2$)}
    \label{fig:amazon2}
  \end{minipage}
\end{figure*}
\begin{figure*}[htbp]
  \centering
  \begin{minipage}[b]{0.28\textwidth}
    \includegraphics[width=\textwidth]{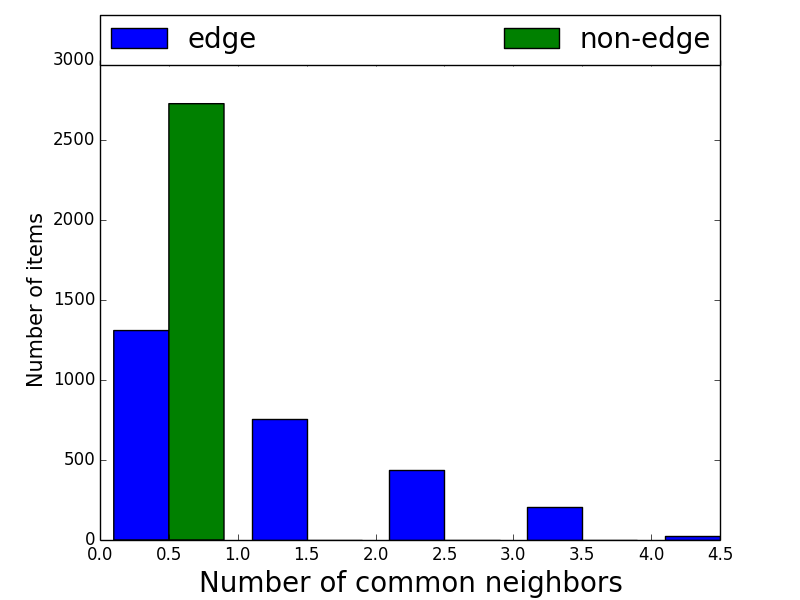}
  \end{minipage}
  \begin{minipage}[b]{0.3\textwidth}
    \includegraphics[width=\textwidth]{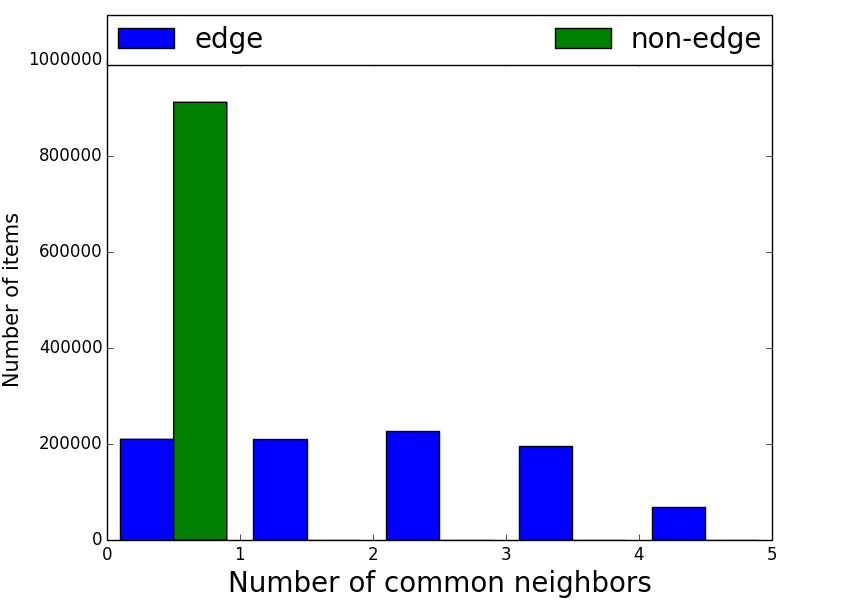}
  \end{minipage}
  \begin{minipage}[b]{0.28\textwidth}
    \includegraphics[width=\textwidth]{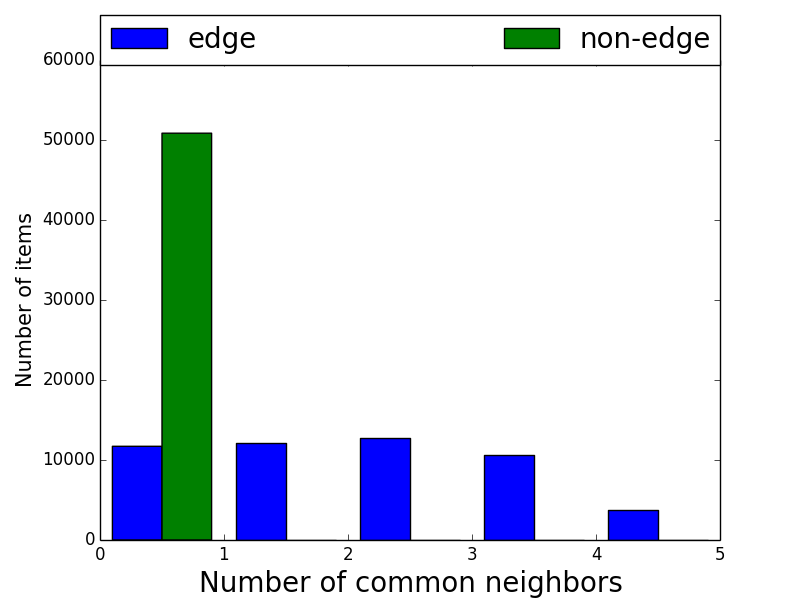}
  \end{minipage}
   \caption{Histogram of common neighbors of edges and non-edges in the co-purchase network, from left to right: Book-DVD, Book-Book, DVD-DVD}
    \label{fig:amazon3}
\end{figure*}
We next take equal number of product pairs inside Book (also inside DVD, and across Book and DVD) that have an edge in-between and do not have an edge respectively. Figure~\ref{fig:amazon3} shows that the number of common neighbors when two products share an edge is much higher than when they do not--in fact, almost all product pairs that do not have an edge in between also do not share any common neighbor. This again strongly suggests towards GBM due to its transitivity property. On the other hand, this also suggests that SBM is not a good model for this network, as in SBM, two nodes having common neighbors is independent of whether they share an edge or not.

{\bf Difference between SBM and GBM.}
It is important to stress that the network structures generated by the SBM and the GBM are quite different, and it is significantly difficult to analyze any algorithm or lower bound on GBM compared to SBM. This difficulty stems from the highly correlated edge generation in GBM (while edges are independent in SBM).  For this reason, analyses of the sphere-comparison algorithm and spectral methods for clustering on GBM cannot be derived as straight-forward adaptations. Whereas, even for simple algorithms, a property that can be immediately seen for SBM, will still require a proof for GBM.

\section{The Motif-Counting Algorithm}
Suppose, we are given a graph $G=(V,E)$ with $2$ disjoint clusters, $V_1, V_2 \subseteq V$ generated according to $GBM(V, t, r_s, r_d)$. 
Our clustering algorithm is based on counting motifs, where a motif is simply defined as a configuration of triplets in the graph. Let us explain this principle by one particular motif, a triangle. For any two vertices $u$ and $v$ in $V$, where $(u,v)$ is an edge, we count the total number of common neighbors of $u$ and $v$. We show that, whenever $r_s\geq 4r_d$, this count is different when $u$ and $v$ belong to the same cluster, compared to when they belong to different clusters. We assume $G$ is connected, because otherwise it is impossible to recover the clusters with certainty. For every pair of vertices in the graph that share an edge, we decide whether they are in the same cluster or not by this count of triangles. In reality, we do not have to check every such pair, instead we can stop when we form a spanning tree. At this point, we can transitively deduce the partition of nodes into clusters.

The main new idea of this algorithm is to use this triangle-count (or motif-count in general), since they carry significantly more information regarding the connectivity of the graph than an edge count. However, we can go to statistics of higher order (such as the two-hop common neighbors) at the expense of increased complexity. 
Surprisingly, the simple greedy algorithm that rely on triplets can separate clusters when $r_d$ and $r_s$ are $\Omega(\frac{\log n}n)$, which is also a minimal requirement for connectivity of random geometric graphs \cite{penrose2003random}. Therefore this algorithm is optimal up to a constant factor. It is interesting to note that this motif-counting algorithm is not optimal for SBM (as we observe), in particular, it will not detect the clusters in the sparse threshold region of $\frac{\log{n}}{n}$, however, it does so for GBM.


The pseudocode of the algorithm is described in Algorithm ~\ref{alg:alg1}. The algorithm looks at individual pairs of vertices to decide whether they belong to the same cluster or not. We go over pair of vertices and label them same/different, till we have enough labels to partition the graphs into clusters. 

At any stage, the algorithm picks up an unassigned node $v$ and queries it with another  node $u: (u,v) \in E$ that has already been assigned to one of the clusters.  Note that it is always possible to find such a vertex $v$ because otherwise the graph would not be connected.
To decide whether these two points $u$ and $v$ belong to the same cluster, the algorithm calls a subroutine named \texttt{process}.  The \texttt{process} function counts the number of common neighbors of $u$ and $v$ to make a decision. The node $v$ is assigned to its respective cluster depending upon the output of \texttt{process} subroutine. This procedure is continued till all nodes in $V$ are assigned to one of the clusters.

\captionof{algorithm}{Cluster recovery in GBM}
\begin{algorithmic}[1]\label{alg:alg1}
{\scriptsize
\REQUIRE GBM $G = (V,E)$, $r_s, r_d$
\ENSURE  $V = V_1\sqcup  V_2$
\STATE Choose any $u\in V$
\STATE $V_1\leftarrow \{u\} $, $V_2 \leftarrow \emptyset$
\WHILE{$ V\neq V_1\sqcup V_2$} 
\STATE Choose $(u,v)\in E \mid u\in V_1\sqcup V_2, v\in V\setminus (V_1\sqcup V_2)$
\IF{{\rm process}($u,v,r_s,r_d$)}
\IF{$u\in V_1$}
\STATE $V_1\leftarrow V_1\cup \{v\}$
\ELSE
\STATE $V_2\leftarrow V_2\cup \{v\}$
\ENDIF
\ELSE
\IF{$u\in V_1$}
\STATE $V_2\leftarrow V_2\cup \{v\}$
\ELSE
\STATE $V_1\leftarrow V_1\cup \{v\}$
\ENDIF
\ENDIF
\ENDWHILE
}
\end{algorithmic}
\captionof{algorithm}{\texttt{process}}
\begin{algorithmic}[1]\label{alg:process}
{\scriptsize
\REQUIRE $u$,$v$, $r_s$, $r_d$
\ENSURE  true/false
\STATE count $\leftarrow |\{z: (z,u)\in E, (z,v)\in E\}|$
\IF{$|\frac{\text{count}}{n} - E_S(r_d,r_s)| < |\frac{\text{count}}{n} -E_D(r_d,r_s)|$}
\RETURN true
\ENDIF
\RETURN false
}
\end{algorithmic}

The \texttt{process} function counts the number of common neighbors of two nodes and then compares the difference of the count with
two functions of $r_d$ and $r_s$, called $E_D$ and $E_S$. 

 \correct{We have compiled the distribution of the number of common neighbors along with other motifs (other patterns of triplets, given $(u,v)\in E$) in Table ~\ref{tab:tab1}. We  provide the values of  $E_D$ and $E_S$  in Theorem \ref{th:main} for the regime of $r_s> 4r_d$.}  In this table we have assumed that there are only two clusters of equal size. The functions change when the cluster sizes are different. Our analysis described in later sections can be used to calculate new function values. In the table, $u\sim v$ means $u$ and $v$ are in the same cluster.

\begin{table*}[htbp]
\begin{center}
\resizebox{\textwidth}{!}{
\begin{tabular}{|c|p{3.2cm}|p{3cm}|p{3cm}|p{3.6cm}|} 
 \hline
 Motif: $(u,v) \in E$ & \multicolumn{2}{c}{Distribution of count ($r_s>2r_d$)}&    \multicolumn{2}{c|}{Distribution of count ($r_s\le 2r_d$)} \\
 $d_L(X_u,X_v) =x$      & $u \sim v,  x\leq r_s$ & $u \nsim v, x\leq r_d$ &  $u \sim v, x\leq r_s$ & $u \nsim v, x\leq r_d$\\
 \hline
Motif 1: $z\mid (z,u)\in E, (z,v)\in E$   &  {\correct{${\rm Bin}(\frac{n}{2}-2,2r_s-x)+\mathbb{1}\{x\leq 2r_d\}{\rm Bin}(\frac{n}{2},2r_d-x)$}} & ${\rm Bin}(n-2,2r_{d})$ &  {\correct{${\rm Bin}(\frac{n}{2}-2,2r_s-x)+{\rm Bin}(\frac{n}{2},2r_d-x)$}}  & \correct{${\rm Bin}(n-2, \min(r_{s}+r_d-x,2r_d))$}\\
 \hline
Motif 2: $z\mid (z,u)\in E, (z,v)\notin E$ & {\correct{${\rm Bin}(\frac{n}{2}-2,x)+{\rm Bin}(\frac{n}{2},\min(x,2r_d))$ }}& \correct{${\rm Bin}(\frac{n}{2}-1,2(r_{s}-r_{d}))  $ }& \correct{${\rm Bin}(n-2, x)$} & \correct{${\rm Bin}(\frac{n}{2}-1,r_s-r_d+x + \max(r_s-x-r_d,0)) + {\rm Bin}(\frac{n}{2}-1,\max(x+r_d-r_s,0)) $ }\\
 \hline
Motif 3:  $z\mid (z,u)\notin E, (z,v)\in E$  &{\correct{${\rm Bin}(\frac{n}{2}-2,x)+{\rm Bin}(\frac{n}{2},\min(x,2r_d))$ }} & \correct{${\rm Bin}(\frac{n}{2}-1,2(r_{s}-r_{d})) $ }& \correct{${\rm Bin}(n-2, x)$} &  \correct{${\rm Bin}(\frac{n}{2}-1,r_s+r_d-x + \max(r_s-x-r_d,0)) + {\rm Bin}(\frac{n}{2}-1,\max(x+r_d-r_s,0)) $ }\\
 \hline
 Motif 4: $z\mid (z,u)\notin E, (z,v)\notin E$  &{\correct{${\rm Bin}(\frac{n}{2}-2,1-(x+2r_s))+\mathbb{1}\{x \le 2r_d\}{\rm Bin}(\frac{n}{2},1-(x+2r_d))+\mathbb{1}\{x > 2r_d\}{\rm Bin}(\frac{n}{2},1-4r_d) $}} & \correct{${\rm Bin}(n-2,1-2r_s) $ }& \correct{${\rm Bin}(\frac{n}{2}-2,1-(x+2r_s))+{\rm Bin}(\frac{n}{2},1-(x+2r_d))$ } &  \correct{$\mathbb{1}\{x \le r_s-r_d\}{\rm Bin}(n-2,1-2r_s)+\mathbb{1}\{x >r_s-r_d\}{\rm Bin}(n-2,1-(x+r_s+r_d))$ }\\
\hline
\end{tabular}}
\end{center}
\caption{Distribution of motif count for an edge $(u,v)$ conditioned on the distance between them $d_L(X_u, X_v) = x$, when there are two equal sized clusters. Here ${\rm Bin}(n,p)$ denotes a binomial random variable with mean $np$.\label{tab:tab1}}
\end{table*}



 Similarly, the \texttt{process} function can be run on other set of motifs by fixing two nodes. On considering a larger set of motifs, the \texttt{process} function can take a majority vote over the decisions received from different motifs.   Note that, our algorithm counts motifs only for edges, and does not count motifs for more than $n-1$ edges, as there are only $n$ vertices to be be assigned to clusters.
 \begin{remark}
If we are given k clusters ($k>2$), our analysis can be extended to calculate new values of $E_S$ and $E_D$. If there exists a palpable gap between the two values, we can extend the Algorithm \ref{alg:alg1} to  identify the true assignment of each node.
 \end{remark}

\section{Analysis of the Algorithm}
~\label{sec:theory}
The critical observation that we have to make to analyze the motif-counting algorithm is the fact that given a GBM graph $G(V,E)$ with two clusters $V = V_1 \sqcup V_2$, and  a pair of vertices $u,v \in V$, the events $\cE^{u,v}_z, z \in V$ of any other vertex $z$ being a common neighbor of 
both $u$ and $v$ given $(u,v) \in E$ are dependent (this is not true in SBM); however given the distance between the corresponding random variables  \correct{$d_L(X_u,X_v) =x$}, the events 
are independent. 
Moreover, the probabilities of $\cE^{u,v}_z\mid (u,v) \in E$ are different when $u$ and $v$ are in the same cluster and when they are in different clusters. Therefore the count of the common neighbors are going to be different, and substantially separated with high probability 
 for two vertices in cases when they are from the same cluster or from different clusters. This will lead the function \texttt{process} to correctly characterize two vertices as being from same or different clusters with high probability.  

Let us now show this more formally. We have the following two lemmas for a GBM graph $G(V,E)$ with two equal-sized (unknown) clusters $V = V_1 \sqcup V_2$, and parameters $r_s,r_d$.

\begin{lemma}\label{lem:indep}
For any two vertices $u,v \in V_i: (u,v) \in E, i =1,2$ belonging to the same cluster,  the event $\cE^{u,v}_z \equiv \{(u,z), (v,z) \in E\}$ is independent with $\cE^{u,v}_w \equiv \{(u,w), (v,w) \in E\}$ conditional on the distance between $X_u$ and $X_v$, $d_L(X_u,X_v) = x$.
\end{lemma}
\begin{proof}
Let us assume that $z,w$ belong to the same cluster as that of $u,v$ (the proof is similar for other cases too, and we omit those cases here). The event $\cE^{u,v}_z  \cap \cE^{u,v}_w$ given $d_L(X_u,X_v) = x$ is equivalent to having  both $X_z$ and $X_w$ (the random variable corresponding to vertices $z$ and $w$ respectively) within a range of $2r_s-x$ if $x\leq 2r_s$ and can never happen if $x>2r_s$. Hence  $\Pr (\cE^{u,v}_z  \cap \cE^{u,v}_w| d_L(X_u,X_v) = x)  = (2r_s-x)^2$ for $x\leq 2r_s$.

On the other hand, the event $\cE^{u,v}_z$  given $d_L(X_u,X_v) = x$ is equivalent to having  $X_z$ within a range of $2r_s-x$ if $x\leq 2r_s$ and 0 otherwise. Similarly the event $\cE^{u,v}_w$ given  $d_L(X_u,X_v) = x$ is equivalent to having $X_w$ within a range of $2r_s-x$.
Therefore $\Pr (\cE^{u,v}_z  \cap \cE^{u,v}_w| d_L(X_u,X_v) = x) = \Pr (\cE^{u,v}_z | d_L(X_u,X_v) = x)\Pr ( \cE^{u,v}_w| d_L(X_u,X_v) = x)$.
\end{proof}

This observation leads to the derivation of distributions of counts of triangles involving $(u,v) \in E$ for the cases when $u$ and $v$ are in the same cluster and when they are not.

\begin{lemma}\label{lem:sep}
For any two vertices $u,v \in V_i: (u,v) \in E, i =1,2$ belonging to the same cluster and  $d_L(X_u,X_v) = x$, the count of common neighbors $C_{u,v} \equiv |\{z\in V: (z,u), (z,v) \in E\}|$ is a random variable distributed  according to ${\rm Bin}(\frac{n}{2}-2,2r_s-x)+{\rm Bin}(\frac{n}{2},2r_d-x)$ if $x \leq \min(2r_d,r_s)$ and 
according to 
${\rm Bin}(\frac{n}{2}-2, 2r_s-x)$ if $2r_d < x \le r_s$, where ${\rm Bin}(n,p)$ is a binomial random variable with mean $np$.
\end{lemma}
\begin{lemma}\label{lem:sep2}
For any two vertices $u\in V_1,v \in V_2: (u,v) \in E$   belonging to different clusters and $d_L(X_u,X_v) = x$, the count of common neighbors $C_{u,v} \equiv |\{z\in V: (z,u), (z,v) \in E\}|$ is a random variable distributed according to ${\rm Bin}(n-2,2r_{d})$ when $r_s > 2r_d$ and according to ${\rm Bin}(n-2,\min(r_{s} + r_d -x,2r_d))$ when $r_s \leq 2r_d$. 
\end{lemma}
 Here let us give the proof of Lemma \ref{lem:sep}. The proof of Lemma \ref{lem:sep2} will follow similarly. These expressions can also be generalized  when the clusters are of unequal sizes, but we omit those for clarity of exposition. 
\begin{proof}[Proof of Lemma \ref{lem:sep}]
Let $X_w\in [0,1]$ be the uniform random variable associated with $w \in V$. Let us also denote by $\dist(X,Y) \equiv \min\{|X - Y|, 1- |X- Y|\}, X,Y \in \reals$. Without loss of generality, assume $u,v \in V_1$. For any vertex $z \in V$, let $\cE^{u,v}_z(x) \equiv \{(u,z), (v,z) \in E \correct{\mid (u,v) \in E, \dist(u,v)=x} \}$ be the event that $z$ is a common neighbor \correct{given that the vertices $u$ and $v$ have an edge and the distance between those vertices is $x$}.
For $z\in V_1$,
\begin{align*}
\Pr(\cE^{u,v}_z(x))  = 2r_s - x, 0 \le x \le r_s.
\end{align*}
For $z \in V_2$,  
\begin{align*}
\Pr(\cE^{u,v}_z(x)) =\begin{cases}
  2r_d - x , \text{ if $x\leq 2r_d$}\\
 0 , \text{ if $2r_d<x\leq r_s$.}
 \end{cases}
\end{align*}
\correct{Since we are conditioning on the fact that the vertices $u$ and $v$ have an edge, $x$ can take a maximum value of $r_s$.}
Now since there are $\frac{n}2-2$ points in $V_1 \setminus \{u,v\}$ and $\frac{n}2$ points in $V_2$, we have the statement of the lemma. 
\end{proof}
The proof of Lemma \ref{lem:sep2} is similar and we delegate it to  Appendix \ref{app:lem3}.

Consider the case when $r_s\ge 4r_d$. The above lemmas show that for all values of $d_L(X_u,X_v)$, the expected count of the number of triangles involving $(u,v)\in E$ is higher when $u$ and $v$ belong to the same cluster as opposed to different clusters. By leveraging the concentration of binomial random variables, we bound the count of the number of triangles in these two cases. We use Lemma \ref{lem:sep} to  first estimate the minimum value of triangle count when $u$ and $v$ belong to the same cluster and Lemma \ref{lem:sep2} to estimate the maximum value of triangle count when $u$ and $v$ belong to different clusters. Our algorithm will correctly resolve whether two points in the same cluster or not  if the minimum value in the former case is higher than the maximum value in the later.
While more general statements are possible, we give a theorem concentrating on the special case when $r_s, r_d \sim \frac{\log n}{n}$, which is at the order of the connectivity threshold of geometric random graphs \cite{penrose2003random}.

\begin{theorem}\label{th:main}
Let $r_s = \frac{a\log{n}}{n}$ and $r_d = \frac{b \log{n}}{n}$,  $a > 4b$, and $g(y) \equiv y+\sqrt{2a-y}+\sqrt{2b-y}.$ Algorithm \ref{alg:alg1} with
$E_D=(2b + \sqrt{6b}) \frac{\log n}n$ and 
$$E_S= \min\left(\frac{a}{2}-\sqrt{a}, a+b-\max_{0\leq \nu \leq 2b}g(\nu) \right)\frac{\log n}n,
$$
 can recover the clusters $V_1, V_2$ accurately  with a probability of $1 -o(1)$ if 
$$\min\left(\frac{a}{2}-\sqrt{a}, a+b-\max_{0\leq \nu \leq 2b}g(\nu) \right) \ge 2b + \sqrt{6b}.$$
\end{theorem}
\begin{proof}
\correct{We need to consider the case of $r_s> 2r_d$ from Lemma \ref{lem:sep} and Lemma \ref{lem:sep2}.}
Let $Z$ denote the random variable that equals the number of common neighbors of two nodes $u,v \in V: (u,v) \in E$. Let us also denote $\mu_s= \avg(Z|u \sim v, \correct{\dist(X_u,X_v)=x)}$ and $\mu_d = \avg(Z|u \nsim v, \correct{\dist(X_u,X_v)=x)}$, where $u\sim v$ means $u$ and $v$ are in the same cluster. 
We can easily find $\mu_s$ and $\mu_d$ from Lemmas \ref{lem:sep}, \ref{lem:sep2}. We see that, 
$\mu_s  =  \begin{cases}
  n(r_s+r_d-x) - 4r_s + 2x , \text{ if $x\leq 2r_d$}\\
 (\frac{n}{2}-2)(2r_s-x) , \text{ if $2r_d<x\leq r_s$}
 \end{cases}
 \text{and } \quad \mu_d = (n-2)2r_d.$
%
%

The value of $\mu_s$ is greater than that of $\mu_d$ for all values of $x$ when $r_s\geq 4r_d$. We try to bound the values of $Z$ in these two cases and then achieve the condition of correct resolution.
Given a fixed $d_L(X_u,X_v)$, since $Z$ is a sum of independent binary random variables, using the Chernoff bound, $\Pr(Z < (1-\delta)\avg(Z)) \leq e^{-\delta^2\avg(Z)/2} = \frac{1}{n\log n},$ when $\delta = \sqrt{\frac{2(\log{n}+\log\log n)}{\avg(Z)}}$. 
Now  when $u,v$ belong to the same cluster and $\dist (X_u,X_v)=x$, with probability at least $1-\frac2{n\log n}$, 
$$ Z \ge F_{\sim}(x) \equiv \begin{cases}
  n(r_s+r_d-x) - 4r_s + 2x - \sqrt{(\log n+\log\log n)(n-4)(2r_s-x)}& \\ \hspace{2in} - \sqrt{(\log n+\log\log n)n(2 r_d-x)} , & \text{ if $x\leq 2r_d$}\\
 (\frac{n}{2}-2)(2r_s-x)  - \sqrt{(\log n+\log\log n)(n-4)(2r_s-x)},& \text{ if $2r_d<x\leq r_s$.}
 \end{cases}
$$ 
Using Chernoff bound, we also know that $\Pr(Z > (1+\delta) \avg(Z)) \le e^{-\delta^2\avg(Z)/3}  = \frac{1}{n\log n},$ when $\delta = \sqrt{\frac{3(\log{n}+\log\log n)}{\avg(Z)}}$.
Hence, with probability at least  $1-\frac1{n\log n}$, $Z$ is at most $F_{\nsim} \equiv \mu_d + \sqrt{3(\log n+\log\log n)\mu_d}$ when $u,v$ belong to different clusters.
 
We calculate the minimum value of $F_\sim(x)$ over all values of $x$ to find the value closest to $F_\nsim$. When $2r_d<x<r_s$, $F_\sim(x)$ is a decreasing function with the minimum value of $(\frac{n}{2}-2)r_s  - \sqrt{(\log n+\log\log n)(n-4)r_s}$ at $x=r_s$. Plugging in $r_s=\frac{a \log n}{n}$, $r_d=\frac{b \log n}{n}$ and $x = \frac{\nu \log n}{n}$ we get that the
 algorithm will be successful to label correctly  with probability $1- \frac{3}{n\log n}$ as long as,
$$
\min(\frac{a}{2}-\sqrt{a}, \min_{0\leq \nu\leq 2b} (a+b- \nu-\sqrt{2a-\nu}-\sqrt{2b-\nu}) ) \log{n} \ge \left(2b + \sqrt{6b}\right)\log{n}.
$$
Now we need the correct assignment of vertices  for $n-1$ pairs of vertices (according to Algorithm \ref{alg:alg1}). Applying union bound over $n-1$ distinct pairs guarantees the probability of recovery as $1-3/\log n$.
\end{proof}


\begin{figure*}[t!]
  \centering
  \begin{subfigure}[t]{0.3\textwidth}
     \includegraphics[height=1.6in]{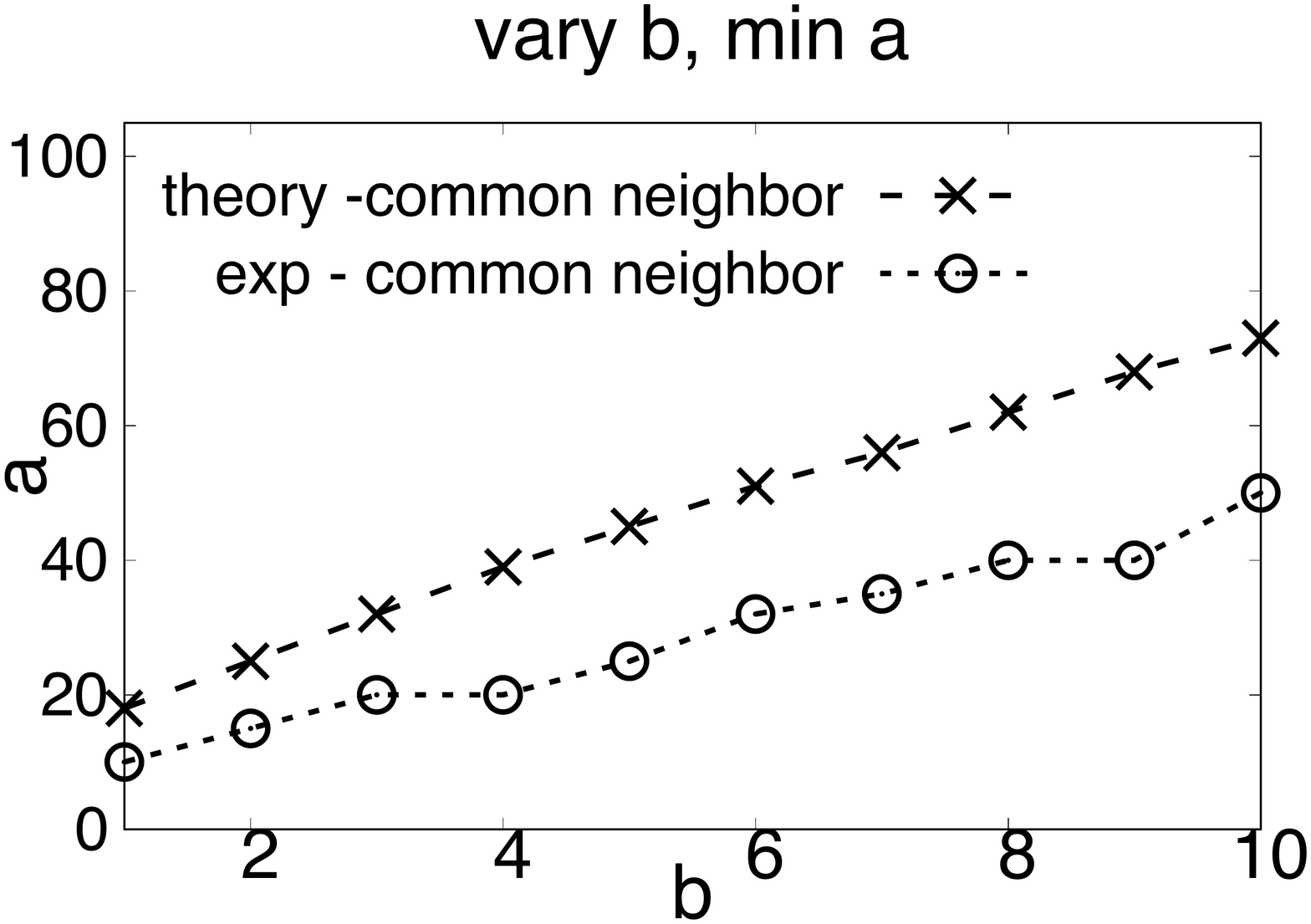}
     \caption{\small Triangle motif varying $b$ and minimum value of $a$ that satisfies the accuracy bound.}
          ~\label{fig:theorypractice}
  \end{subfigure}%
 \begin{subfigure}[t]{0.3 \textwidth}
   \includegraphics[height=1.6in]{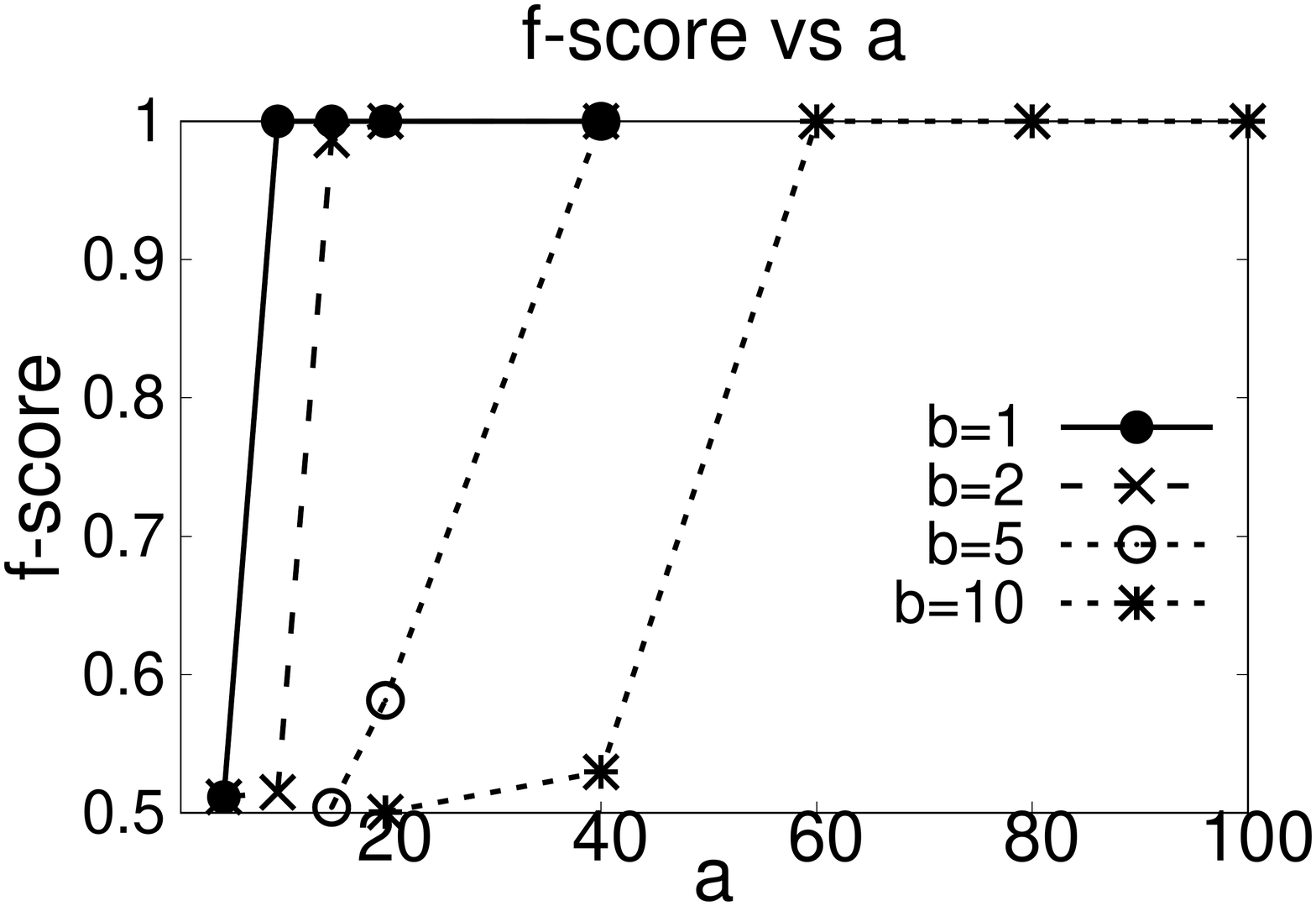}
   \caption{\small f-score with varying $a$, fixing $b$.}
       ~\label{fig:varya}
 \end{subfigure}%
 \begin{subfigure}[t]{0.3\textwidth}
    \includegraphics[height=1.6in]{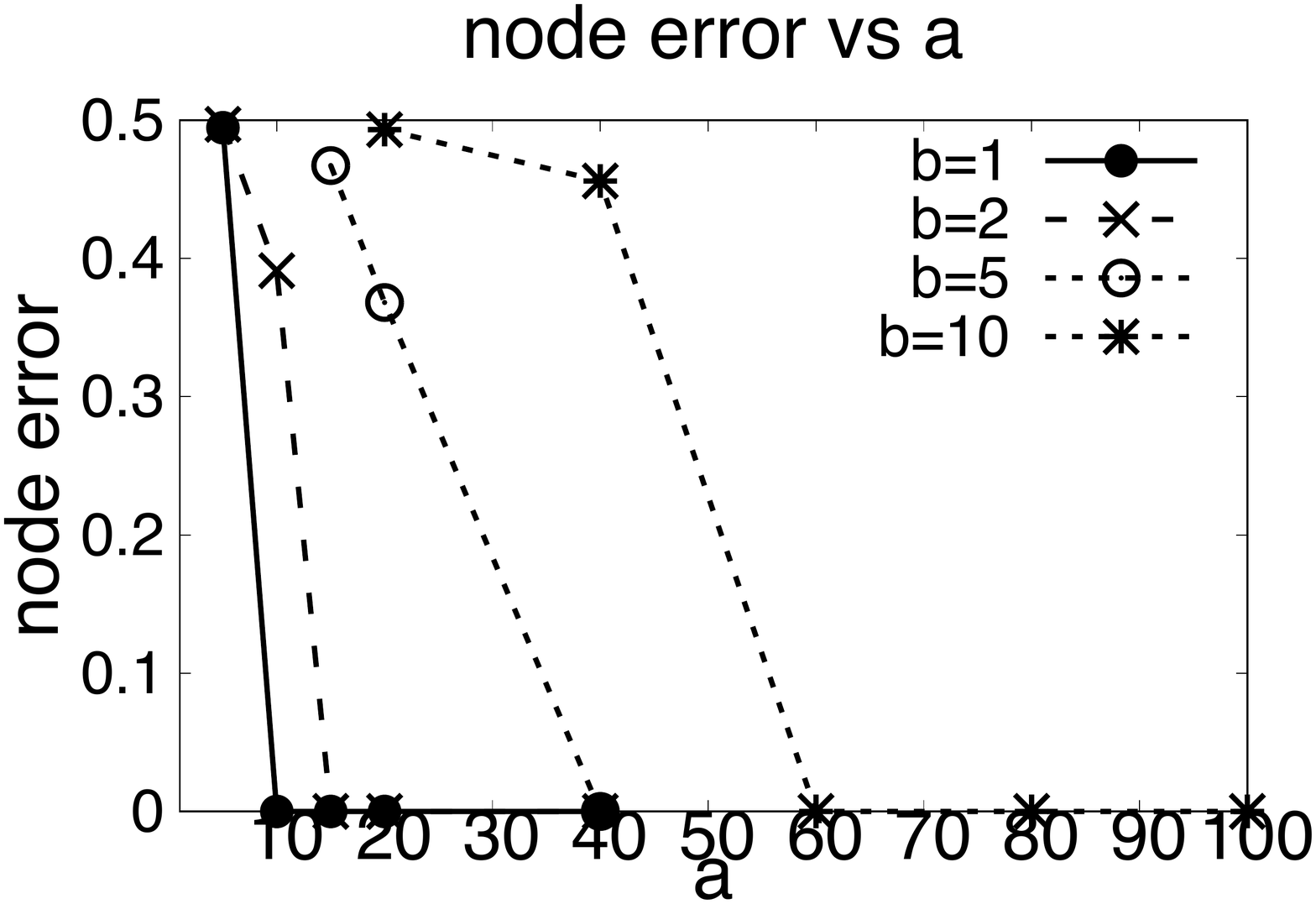}
    \caption{Fraction of nodes misclassified.}
          ~\label{fig:node}
  \end{subfigure}
 \caption{\small Results of the motif-counting algorithm on a synthetic dataset with $5000$ nodes.}
\end{figure*}

\begin{table*}[htbp]
\centering
\small{
 \resizebox{\textwidth}{!}{\begin{tabular}{||c| c| c| c| c| c| c|c|c||} 
 \hline
 Dataset &Total no.  & $T_1$ & $T_2$ & $T_3$ &\multicolumn{2}{c|}{Accuracy}  & \multicolumn{2}{c||}{Running Time (sec)}\\ 
   &of nodes &  &  &  &Motif-Counting & Spectral clustering & Motif-Counting&Spectral clustering\\ [0.5ex] 
\hline\hline
Political Blogs & 1222 & 20 & 2 & 1&  \textbf{0.788}& 0.53 & 1.62 & \textbf{0.29} \\
DBLP & 12138    & 10 & 1 & 2& \textbf{0.675}&0.63 & \textbf{3.93}&18.077\\
LiveJournal & 2366    & 20 & 1 & 1& \textbf{0.7768}&0.64 & \textbf{0.49} & 1.54\\ 
 \hline
\end{tabular}}}
\caption{\small Performance on real world networks}
\label{tab:realtab}
\end{table*}

Instead of relying only on the triangle (or common-neighbor) motif, we can consider other different motifs (as listed in Table \ref{tab:tab1}) and use them to make similar analysis. Aggregating the different motifs by taking a majority vote decision may improve the results experimentally but it is difficult to say anything theoretically since the decisions of the different motifs are not independent. We refer the reader to  Appendix \ref{ap:motifs} for the detailed analysis of incorporating other motifs to obtain analogous theorems. 

\begin{remark}
Instead of using Chernoff bound we could have used better concentration inequality (such as Poisson approximation) in the above analysis, to get tighter condition on the constants. We again preferred to keep things simple.
\end{remark}

\begin{remark}[GBM for $t =3$ and above]
For GBM with $t=3$, to find the number of common neighbors of two vertices, we need to find out the area of intersection of two spherical caps on the sphere. It is possible to do that. It can be seen that, our algorithm will successfully identify the clusters as long as $r_s, r_d \sim \sqrt{\frac{\log n}{n}}$ again when the constant terms satisfy some conditions. However tight characterization becomes increasingly difficult. For general $t$, our algorithm should be successful when $r_s, r_d \sim \Big(\frac{\log n}{n}\Big)^{\frac{1}{t-1}}$, which is also the regime of connectivity threshold. 
\end{remark}
\begin{remark}[More than two clusters] When there are more than two clusters, the same analysis technique is applicable and we can estimate the expected number of common neighbors. This generalization can be  straightforward but tedious. 
\end{remark}

{\bf Motif counting algorithm for SBM.}
While our algorithm is near optimal for GBM in the regime of $r_s, r_d \sim \frac{\log n}{n}$, it is far from optimal for the SBM in the same regime of average degree. Indeed, by using simple Chernoff bounds again, we see that the motif counting algorithm is successful for SBM with inter-cluster edge probability $q$ and intra-cluster probability $p$, when $p,q \sim \sqrt{\frac{\log n}{n}}$. The experimental success of our algorithm in real sparse networks therefore somewhat enforce the fact that GBM is a better model for those network structures than SBM.


\section{Experimental Results}\label{sec:exp}
In addition to validation experiments in Section~\ref{subsec:m1} and~\ref{subsec:m2}, we also conducted an in-depth experimentation of our proposed model and techniques over a set of synthetic and real world networks. Additionally, we compared the efficacy and efficiency of our motif-counting algorithm with the popular spectral clustering algorithm using normalized cuts\footnote{http://scikit-learn.org/stable/modules/clustering.html\#spectral-clustering} and the correlation clustering algorithm~\cite{bbc:04}.

\noindent{\bf Real Datasets.} We use three real datasets described below.

\begin{itemize}[noitemsep,leftmargin=5pt]
\item \textbf{Political Blogs.}~\cite{politicalblog} It contains a list of political blogs from 2004 US Election classified as liberal or conservative, and links between the blogs. The clusters are of roughly the same size with a total of 1200 nodes and 20K edges.

\item \textbf{DBLP.}~\cite{communityreal} The DBLP dataset is a collaboration network where the ground truth communities are defined by the research community. The original graph consists of roughly 0.3 million nodes. We process it to extract the top two communities of size $\sim$ 4500 and 7500 respectively. This is given as input to our algorithm.

\item \textbf{LiveJournal.}~\cite{livejournal} The LiveJournal dataset is a free online blogging social network of around 4 million users. Similar to DBLP, we extract the top two clusters of sizes 930 and 1400 which consist of around 11.5K edges.
\end{itemize}
  We have not used the academic collaboration (Section~\ref{subsec:m1}) dataset here because it is quite sparse and below the connectivity threshold regime of both GBM and SBM.

\noindent{\bf Synthetic Datasets.} We generate synthetic datasets of different sizes according to the GBM with $t=2, k=2$ and for a wide spectrum of  values of $r_s$ and $r_d$, specifically we focus on the sparse region where $r_s = \frac{a\log{n}}{n}$ and $r_d = \frac{b\log{n}}{n}$ with variable values of $a$ and $b$.

\noindent{\bf Experimental Setting.} For real networks,  it is difficult to calculate an exact threshold as the exact values of $r_s$ and $r_d$ are not known. Hence, we follow a three step approach. Using a somewhat large threshold $T_1$ we sample a subgraph $S$ such that $u,v$ will be in $S$ if there is an edge between $u$ and $v$, and they have at least $T_1$ common neighbors. We now attempt to recover the subclusters inside this subgraph by following our algorithm with a small threshold $T_2$. Finally, for nodes that are not part of $S$, say $x \in V \setminus S$, we select each $u \in S$ that  $x$ has an edge with and use a threshold of $T_3$ to decide if $u$ and $x$ should be in the same cluster. The final decision is made by taking a majority vote. We can employ sophisticated methods over this algorithm to improve the results further, which is beyond the scope of this work. 

We use the popular f-score metric which is the harmonic mean of precision (fraction of number of pairs correctly classified to total number of pairs classified into clusters) and recall (fraction of number of pairs correctly classified to the total number of pairs in the same cluster for ground truth), as well as the node error rate for performance evaluation. A node is said to be misclassified if it belongs to a cluster where the majority comes from a different ground truth cluster (breaking ties arbitrarily). Following this, we use the above described metrics to compare the performance of different techniques on various datasets.

\noindent{\bf Results.} We  compared our algorithm with the spectral clustering algorithm where we extracted two eigenvectors in order to extract two communities. Table~\ref{tab:realtab} shows that our algorithm gives an accuracy as high as $78\%$. The spectral clustering  performed worse  compared to our algorithm for all real world datasets. It obtained the worst accuracy of 53\% on political blogs dataset. The correlation clustering algorithm generates various small sized clusters leading to a very low recall, performing much worse than the motif-counting algorithm for the whole spectrum of parameter values.

We can observe in Table~\ref{tab:realtab} that our algorithm is much faster than the spectral clustering algorithm for larger datasets (LiveJournal and DBLP). This  confirms that motif-counting algorithm is more scalable than the spectral clustering algorithm.  The spectral clustering algorithm also works very well on synthetically generated SBM networks even in the sparse regime \cite{lei2015consistency,rohe2011spectral}. The superior performance of the simple motif clustering algorithm over the real networks provide a further validation of GBM over SBM. Correlation clustering  takes 8-10 times longer as compared to motif-counting algorithm for the various range of its parameters.  We also compared our algorithm with the Newman algorithm \cite{newman} that performs really well for the LiveJournal dataset (98\% accuracy). But it is extremely slow and performs much worse on other datasets. This is because the LiveJournal dataset has two well defined subsets of vertices with very few intercluster edges. The reason for the worse performance of our algorithm is the sparseness of the graph. If we create a subgraph by removing all nodes of degrees $1$ and $2$, we get $100\%$ accuracy with our algorithm. Finally, our algorithm is easily parallelizable to achieve better improvements. This clearly establishes the efficiency and effectiveness of motif-counting.

\begin{figure}[t]
  \centering
   \begin{subfigure}[c]{0.5 \textwidth}
   \includegraphics[width=\textwidth]{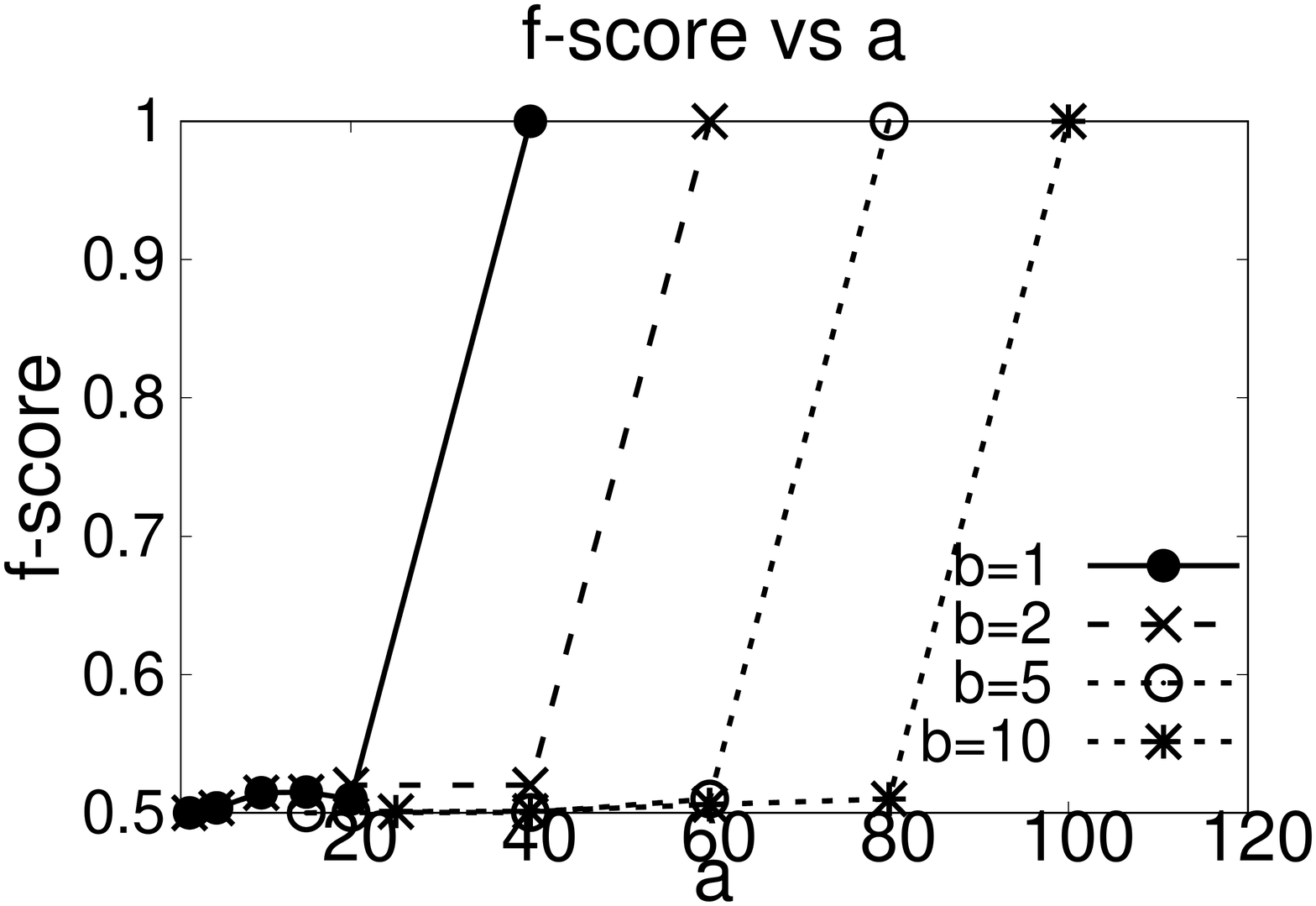}
    \caption{f-score with varying $a$, fixed $b$.}
      ~\label{fig:varya1}
 \end{subfigure}%
 \begin{subfigure}[c]{0.5\textwidth}
    \includegraphics[width=\textwidth]{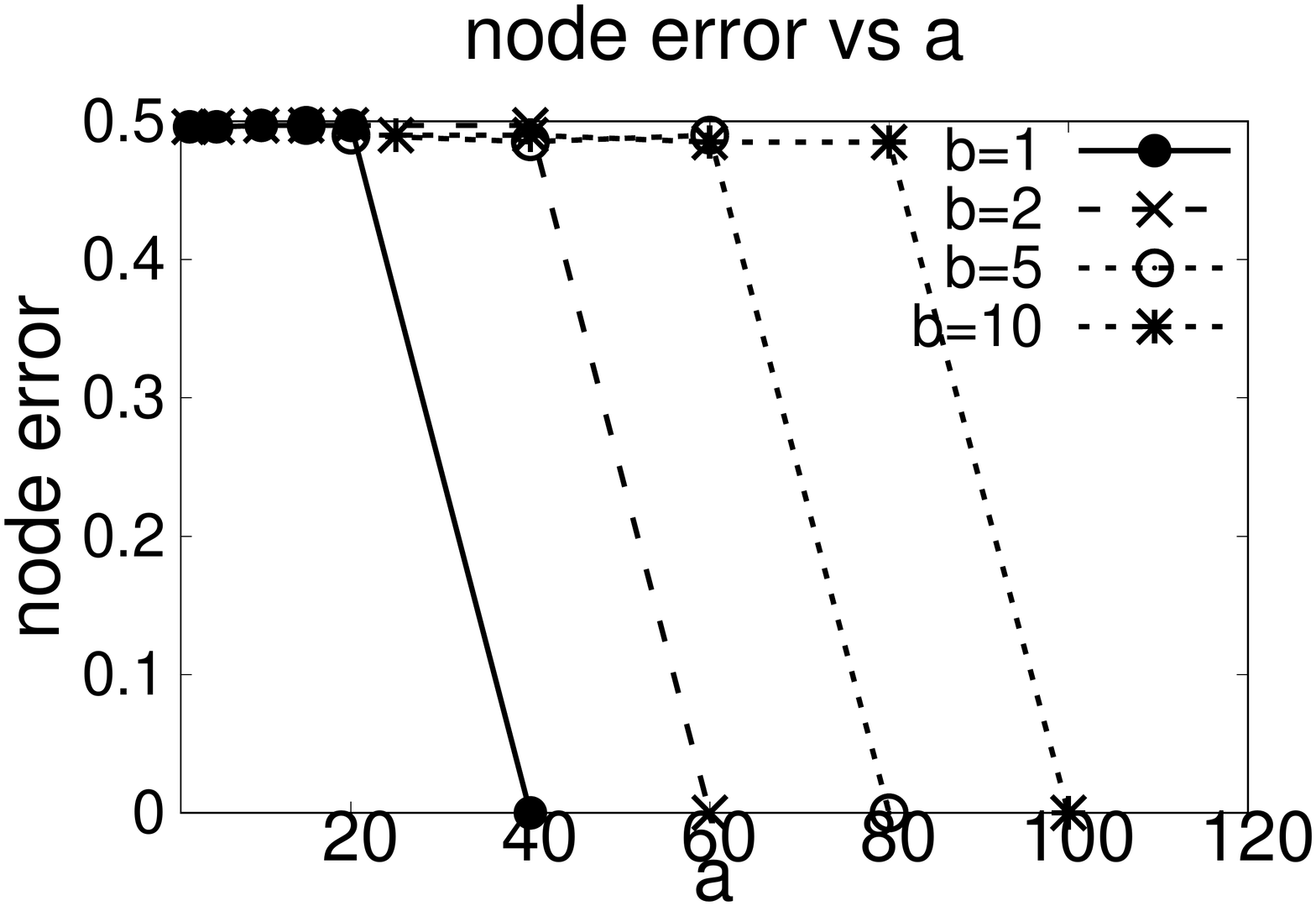}
     \caption{Fraction of  nodes misclassified.}
     ~\label{fig:fscore-1}
  \end{subfigure}
   \vspace{-0.1in}
 \caption{\small Results of the spectral clustering on a synthetic dataset with $5000$ nodes.}
\end{figure}

We observe similar gains on synthetic datasets. 
Figures~\ref{fig:theorypractice}, \ref{fig:varya} and \ref{fig:node} report results on the synthetic datasets with $5000$ nodes. Figure~\ref{fig:theorypractice} plots the minimum gap between $a$ and $b$ that guarantees exact recovery according to Theorem~\ref{th:main}  vs minimum value of $a$ for varying $b$ for which experimentally (with only triangle motif) we were able to recover the clusters exactly. Empirically, our results demonstrate \correct{much superior performance of our algorithm. The empirical results are much better than the theoretical bounds because the concentration inequalities applied in Theorem~\ref{th:main} assume the worst value of the distance between the pair of vertices that are under consideration.}  We also see a clear threshold behavior on both f-score and node error rate in Figures~\ref{fig:varya} and \ref{fig:node}. We have also performed spectral clustering on this 5000-node synthetic dataset (Figures~\ref{fig:varya1} and \ref{fig:fscore-1}). Compared to the plots of figures \ref{fig:varya} and \ref{fig:node}, they show suboptimal performance, indicating the relative ineffectiveness of spectral clustering in GBM compared to the motif counting algorithm.

\bibliographystyle{abbrv}

\bibliography{bibfile}

\section{Appendix}

For the analysis, let $X_w\in [0,1]$ be the uniform random variable associated with $w \in V$. Also recall that $\dist(X,Y) \equiv \min\{|X - Y|, 1- |X- Y|\}, X,Y \in \reals$. 

\subsection{Proof of Lemma \ref{lem:sep2}}\label{app:lem3}

\begin{proof}[Proof of Lemma \ref{lem:sep2}]
Here $u,v$ are from different clusters. For any vertex $z \in V$, let $\cE^{u,v}_z(x) \equiv \{(u,z), (v,z) \in E\mid (u,v)\in E, \dist(X_u,X_v)=x\}$ be the event that $z$ is a common neighbor.
For $z\in V\setminus\{u,v\}$,
\begin{align*}
\Pr(\cE^{u,v}_z(x)) &= \Pr( (z,u) \in E, (z,v) \in E \mid (u,v) \in E, \dist(X_u,X_v)=x)\\
& = \min\{2r_d, r_d+r_s-x\} \\
& = \begin{cases}2r_d & \text{ if } 2r_d < r_s\\ r_s +r_d - x & \text{ otherwise} \end{cases} .
\end{align*}
Now since there are $n-2$ points in $V \setminus \{u,v\}$, we have the statement of the lemma. 
\end{proof}

\subsection{Results for other motifs} \label{ap:motifs}

Next, we describe two lemmas for a GBM graph $G(V,E)$ with two unknown clusters $V=V_1\sqcup V_2$, and parameters $r_s$,$r_d$ on considering  other motifs than triangles (Motif 1). These results are used to populate Table \ref{tab:tab1}. When we run Algorithm \ref{alg:alg1} with other motifs, the subroutine \texttt{process} uses the corresponding motifs to compute the variable `count'. Other than this the algorithm remains same.

\subsubsection{Motif 2 amd Motif 3}

\begin{lemma}
For any two vertices $u,v \in V_i: (u,v) \in E, i =1,2$ belonging to the same cluster and $\dist(X_u,X_v)=x$, the count of number of nodes forming Motif 2 (see Table \ref{tab:tab1}) with $u$ and $v$ (i.e., neighbors  of $u$ and non neigbors of $v$), $|\{z\in V: (z,u)\in E, (z,v) \notin E\}|$ is a random variable distributed according to ${\rm Bin}(\frac{n}{2}-2,x)+{\rm Bin}(\frac{n}{2},\min(2r_d,x))$, where ${\rm Bin}(n,p)$ is a binomial random variable with mean $np$.\label{lem:sep3}
\end{lemma}
\begin{proof}
Without loss of generality, assume $u,v \in V_1$. For any vertex $z \in V$, let $\cE^{u,v}_z(x) = \{(u,z)\in E, (v,z) \notin E\mid (u,v)\in E, \dist(X_u,X_v)=x\}$ be the event that $z$ is a neighbor of $u$ and non neighbor of $v$.
For $z\in V_1$,
\begin{align*}
\Pr(\cE^{u,v}_z(x)) &= \Pr( (z,u) \in E, (z,v) \notin E \mid (u,v) \in E, \dist(X_u,X_v)=x)= r_s-(r_s-x) =x.
\end{align*}
For $z \in V_2$, we have, 
\begin{align*}
\Pr(\cE^{u,v}_z(x)) &= \Pr( (z,u) \in E, (z,v) \notin E \mid (u,v) \in E,  \dist(X_u,X_v)=x)\\
&= \begin{cases}2r_d & \text{ if } x>2r_d\\ x & \text{ otherwise} \end{cases} .
\end{align*}
Now since there are $\frac{n}2-2$ points in $V_1 \setminus \{u,v\}$ and $\frac{n}2$ points in $V_2$, we have the statement of the lemma. 
\end{proof}

\begin{lemma}
For any two vertices $u\in V_1,v \in V_2: (u,v) \in E$ belonging to different clusters and $\dist(X_u,X_v)=x$, the count of number of nodes forming Motif 2 (see Table \ref{tab:tab1}) with $u$ and $v$ (i.e. neighbor of $u$ and non neighbor of $v$),  $|\{z\in V: (z,u)\in E, (z,v) \notin E\}|$ is a random variable distributed according to ${\rm Bin}(\frac{n}{2}-1,2(r_s - r_{d}))$, assuming $r_s>2r_d$.\label{lem:sep4}
\end{lemma}
\begin{proof}

For any vertex $z \in V_{1}$, let $\cE^{u,v}_z(x) = \{(u,z)\in E, (v,z) \notin E \mid (u,v)\in E, \dist(X_u,X_v)=x\}$ be the event that $z$ is a  neighbor of $u$ and a non neighbor of $v$.
For $z\in V_{1}\setminus\{u\}$
\begin{align*}
\Pr(\cE^{u,v}_z(x)) &= \Pr( (z,u) \in E, (z,v) \notin E \mid (u,v) \in E, \dist(X_u,X_v)=x)\\
& = 2(r_s-r_d).
\end{align*}
Now for $z\in V_{2}\setminus\{v\}$, there cannot be an edge with $u$ and no edge with $v$ because $r_s>2r_d$.
Since there are $\frac{n}{2}-1$ points in $V_{1} \setminus \{u\}$, we have the statement of the lemma.
\end{proof}

\begin{theorem}[Motif 2 or 3]\label{th:main2}
If $r_s = \frac{a\log{n}}{n}$ and $r_d = \frac{b \log{n}}{n}$,  $a > 4b$, Algorithm \ref{alg:alg1} with 
$E_S = \left(b + \frac{a}{2} + \sqrt{3b}+ \sqrt{\frac{3a}{2}}\right)\frac{\log n}{n}$ and $E_D =\left((a-b) - \sqrt{2(a-b)}\right)\frac{\log n}{n}$, where the \texttt{process}
subroutine counts Motif 2,
can recover the clusters $V_1, V_2$ accurately  with a probability of at least $1 -o(1)$ if 
$$(a-b) - \sqrt{2(a-b)}  > b + \frac{a}{2} + \sqrt{3b}+ \sqrt{\frac{3a}{2}}. $$
\end{theorem}
\begin{proof}
\correct{We need to consider the case of $r_s\geq 2r_d$ from Lemma \ref{lem:sep3} and Lemma \ref{lem:sep4}.}
For $u,v \in V: (u,v) \in E$, let $Z$ denote the random variable that equals the number of nodes that are neighbors of $u$ and not-a-neighbor of $v$. Let us also denote $\mu_s= \avg(Z|u \sim v, \correct{\dist(X_u,X_v)=x)}$ and $\mu_d = \avg(Z|u \nsim v, \correct{\dist(X_u,X_v)=x)}$, where $u\sim v$ means $u$ and $v$ are in the same cluster. 
We can easily find $\mu_s$ and $\mu_d$ from Lemmas \ref{lem:sep3}, \ref{lem:sep4}. We see that, 
$$\mu_s  =  \begin{cases}
  (n-2)x, \text{ if $x\leq 2r_d$}\\
( \frac{n}{2}-2)x + nr_d , \text{ if $2r_d<x\leq r_s$.}
 \end{cases}
 \text{and } \quad \mu_d = (n-2)(r_s-r_d).$$
%
%

The value of $\mu_s$ is less than that of $\mu_d$ for all values of $x$ when $r_s\geq 4r_d$. We try to bound the values of $Z$ in the two cases possible and then achieve the condition of correct resolution.
Given a fixed $d_L(X_u,X_v)$, since $Z$ is a sum of independent binary random variables, using the Chernoff bound, $\Pr(Z < (1-\delta)\avg(Z)) \leq e^{-\delta^2\avg(Z)/2} = \frac{1}{n\log n},$ when $\delta = \sqrt{\frac{2(\log{n}+\log\log n)}{\avg(Z)}}$. 
Now with probability at least $1-\frac1{n\log n}$,  Z is atleast $ F^1_{\nsim}(x) \equiv \mu_d - \sqrt{2(\log n+\log\log n)\mu_d}$ when $u$ and $v$ belong to different clusters.

Using Chernoff bound, we also know that $\Pr(Z > (1+\delta) \avg(Z)) \le e^{-\delta^2\avg(Z)/3}  = \frac{1}{n\log n},$ when $\delta = \sqrt{\frac{3(\log{n}+\log\log n)}{\avg(Z)}}$.
Hence, with probability at least  $1-\frac2{n\log n}$,

$$ Z \le  F^1_{\sim}(x) \equiv \begin{cases}
  (n-2)x +  \sqrt{3(\log n+\log\log n)\frac{n}{2}x} + \sqrt{3(\log n+\log\log n)(\frac{n}{2}-2)x}  , \text{ if $x\leq 2r_d$}\\
  nr_d + (\frac{n}{2}-2)x +  \sqrt{3(\log n+\log\log n)nr_d} & \\ \hspace{2in} + \sqrt{3(\log n+\log\log n)(\frac{n}{2}-2)x}  , \text{ if $2r_d<x\leq r_s$.}
 \end{cases}
$$
 when $u,v$ belong to the same cluster and $\dist (X_u,X_v)=x$. 
 
We calculate the maximum value of $F^1_{\sim}(x)$ over all values of $x$ to find the value closest to $F^1_{\nsim}(x)$. $F^1_{\sim}(x)$ is an increasing function $\forall x\leq r_s$ with maximum value of $nr_d+(\frac{n}{2}-2)r_s +  \sqrt{3(\log n+\log\log n)nr_d} + \sqrt{3(\log n+\log\log n)(\frac{n}{2}-2)r_s}$   at $x=r_s$. Therefore the algorithm will be successful to label correctly  with probability $1- \frac{3}{n\log n}$ as long as,
$$
\left((a-b) - \sqrt{2(a-b)}\right) \log{n} \ge \left(b + \frac{a}{2} + \sqrt{3b}+ \sqrt{\frac{3a}{2}}\right) \log{n}.
$$ 
The rest of the argument follows similar to Theorem \ref{th:main}.
\end{proof}

\remove{
\latest{
Following the proof of Theorem \ref{th:main2}, we can calculate the values of $E_S$ and $E_D$ for the second type of motif as defined by the following lemma.
\begin{lemma}\label{lem:esed2}
If $r_s = \frac{a\log{n}}{n}$ and $r_d = \frac{b \log{n}}{n}$,  $a > 4b$, Algorithm \ref{alg:alg1} can recover the clusters $V_1, V_2$ accurately using the common neighbor motif with a probability of $1 -o(1)$ if 
$\left((a-b) - \sqrt{2(a-b)}\right)  \ge \left(b + \frac{a}{2} + \sqrt{3b}+ \sqrt{\frac{3a}{2}}\right)$, with $E_D=r_s-r_d$ and 
$E_S=\frac{r_s}{2}+r_d$
\end{lemma}
}}

\subsubsection{Motif 4}

In this part we are concerned with the motif where for $(u,v)\in E$, we seek nodes that are neighbors of neither $u$ nor $v$.

\begin{lemma}
For any two vertices $u,v \in V_i: (u,v) \in E, i =1,2$ belonging to the same cluster and $\dist(X_u,X_v)=x$, the count of nodes that form Motif 4 with $u,v$ (i.e.,  non-neighbors  of both $u$ and $v$), $|\{z\in V: (z,u)\notin E, (z,v) \notin E\}|$ is a random variable distributed according to ${\rm Bin}(\frac{n}{2}-2,1-(x+2r_s))+\mathbb{1}\{x \le 2r_d\}{\rm Bin}(\frac{n}{2},1-(x+2r_d))+\mathbb{1}\{x > 2r_d\}{\rm Bin}(\frac{n}{2},1-4r_d) $, when $r_s>2r_d$. 
\label{lem:sep5}
\end{lemma}
\begin{proof}
Without loss of generality, assume $u,v \in V_1$. For any vertex $z \in V$, let $\cE^{u,v}_z(x) = \{(u,z)\notin E, (v,z) \notin E\}$ be the event that $z$ is neither a neighbor of $u$ nor a neighbor of $v$.
For $z\in V_1$,
\begin{align*}
\Pr(\cE^{u,v}_z(x)) &= \Pr( (z,u) \notin E, (z,v) \notin E \mid (u,v) \in E, \dist(X_u,X_v)=x)\\
& = 1 - \Pr( (z,u) \in E \text{ or } (z,v) \in E \mid (u,v) \in E, \dist(X_u,X_v)=x)\\
& = 1 - (x+2r_s).
\end{align*}
For $z \in V_2$, we have, 
\begin{align*}
\Pr(\cE^{u,v}_z(x)) &= \Pr( (z,u) \notin E, (z,v) \notin E \mid (u,v) \in E,  \dist(X_u,X_v)=x)\\
& = 1 - \Pr( (z,u) \in E \text{ or } (z,v) \in E \mid (u,v) \in E, \dist(X_u,X_v)=x)\\
&= \begin{cases}1 - (2r_d  +x)& \text{ if } x\leq 2r_d\\ 1- 4r_d & \text{ otherwise} \end{cases} .
\end{align*}
Now since there are $\frac{n}2-2$ points in $V_1 \setminus \{u,v\}$ and $\frac{n}2$ points in $V_2$, we have the statement of the lemma. 
\end{proof}

\begin{lemma}
For any two vertices $u\in V_1,v \in V_2: (u,v) \in E$ belonging to different clusters and $\dist(X_u,X_v)=x$, the count of number of nodes forming Motif 4 with $u$ and $v$ (i.e. non-neighbor of $u$ and non-neighbor of $v$),  $ |\{z\in V: (z,u)\notin E, (z,v) \notin E\}|$ is a random variable distributed according to ${\rm Bin}(n-2,1-2r_s)$,   when $r_s>2r_d$.\label{lem:sep6}
\end{lemma}
\begin{proof}
For any vertex $z \in V_{1}$, let $\cE^{u,v}_z(x) = \{(u,z)\notin E, (v,z) \notin E \mid (u,v)\in E, \dist(X_u,X_v)=x\}$ be the event that $z$ is neither a  neighbor of $u$ and nor a neighbor of $v$.
For $z\in V_{1}\setminus\{u\}$
\begin{align*}
\Pr(\cE^{u,v}_z(x)) &= \Pr( (z,u) \notin E, (z,v) \notin E \mid (u,v) \in E, \dist(X_u,X_v)=x)\\
& = 1-2r_s.
\end{align*}

Similarly, for $z\in V_{2}\setminus\{v\}$
\begin{align*}
\Pr(\cE^{u,v}_z(x)) &= \Pr( (z,u) \notin E, (z,v) \notin E \mid (u,v) \in E, \dist(X_u,X_v)=x)\\
& = 1-2r_s.
\end{align*}

Now since there are $\frac{n}{2}-1$ points in $V_{1} \setminus \{u\}$ and $\frac{n}{2}-1$ points in $V_{2} \setminus \{v\}$, we have the statement of the lemma.
\end{proof}

It turns out that the simple Chernoff bound is not sufficient to prove any meaningful result for this motif. We recall the Bernstein's inequality in Lemma \ref{lem:bern} in order to prove Theorem \ref{th:main3} for the 4th motif.
\begin{lemma}[Bernstein's Inequality \cite{boucheron2004concentration}]
Let $X_1,\dots,Xn$ be iid real-valued random variables with mean zero, such that $|X_{i}| \le M \; \forall i$. Then with probability at least $1-\delta$, we have 
$$ \mid\sum_{i=1}^{n}X_{i}\mid  \le \sqrt{2n\mathbb{E}X_{1}^{2}\log \frac{2}{\delta}}+\frac{2M\log \frac{2}{\delta}}{3}.$$\label{lem:bern}
\end{lemma}

\begin{theorem}[Motif 4]\label{th:main3}
If $r_s = \frac{a\log{n}}{n}$ and $r_d = \frac{b \log{n}}{n}$,  $a > 4b$, Algorithm \ref{alg:alg1} with 
$$E_S=1-\frac32r_s-2r_d-\sqrt{3r_s\frac{\log n}{n}}- \sqrt{4r_d\frac{\log n}{n}} - \frac{4\log n}{3n}$$ and $$E_D=1-2r_s+\sqrt{4r_s\frac{\log n}{n}}+\frac{2\log n}{3n},$$
where the \texttt{process}
subroutine counts Motif 4,
 can recover the clusters $V_1, V_2$ accurately with a probability of $1 -o(1)$ if 
$\mid a-4b \mid \ge 2(\sqrt{3a}+\sqrt{4b}+\sqrt{4a}+2)$.
\end{theorem}
\begin{proof}
\correct{We need to consider the case of $r_s\geq 2r_d$.}
Let $Z$ denote the random variable that equals the number of common non-neighbors of two nodes $u,v \in V: (u,v) \in E$. Let us also denote $\mu_s= \avg(Z|u \sim v, \correct{\dist(X_u,X_v)=x)}$ and $\mu_d = \avg(Z|u \nsim v, \correct{\dist(X_u,X_v)=x)}$, where $u\sim v$ means $u$ and $v$ are in the same cluster. 
We can easily find $\mu_s$ and $\mu_d$ from Lemmas \ref{lem:sep5}, \ref{lem:sep6}. We see that, 
$\mu_s  =  \begin{cases}
(\frac{n}{2} -2)(1-x-2r_s) + \frac{n}{2}(1-x-2r_d), \text{ if $x\leq 2r_d$}\\
(\frac{n}{2} -2)(1-x-2r_s) + \frac{n}{2}(1-4r_d), \text{ if $2r_d<x\leq r_s$.}
 \end{cases}
 \text{and } \quad \mu_d = (n-2)(1-2r_s).$\\
%
%

The value of $\mu_s$ is more than that of $\mu_d$ for all values of $x$ when $r_s\geq 4r_d$. We try to bound the values of $Z$ in the two cases possible and then achieve the condition of correct resolution.
Now we will use Bernstein's inequality as defined in Lemma \ref{lem:bern}.

For a Bernoulli($p$) random variable $X$ we can define a corresponding zero mean random variable $\hat{X} \equiv X-\expect{X}$. It is easy to observe that $\expect{\hat{X}}^{2} =p(1-p) \le 1-p$ and $|\hat{X}| \le 1$. We use this simple translation for every random variable corresponding to each node forming such a motif with $u$ and $v$ and hence with a probability of at least $1-\frac{2}{n\log n}$, we must have
$$ Z \ge F^2_{\sim}(x) \equiv \begin{cases}
   (\frac{n}{2} -2)(1-x-2r_s) + \frac{n}{2}(1-x-2r_d)- \sqrt{n(x+2r_s)(\log2n+\log \log n)}& \\  \hspace{0.5in}- \sqrt{n(x+2r_d)(\log2n+\log \log n)}-\frac{4 (\log2n+\log \log n)}{3} , & \text{ if $x\leq 2r_d$}\\
 (\frac{n}{2} -2)(1-x-2r_s) + \frac{n}{2}(1-4r_d)- \sqrt{n(x+2r_s)(\log 2n+\log \log n)}& \\  - \sqrt{4nr_d(\log 2n+\log \log n)}-\frac{4( \log2n+\log \log n)}{3}, & \text{ if $2r_d<x\leq r_s$,}
 \end{cases}
$$ 
when $u$ and $v$ are in the same cluster.
Similarly, with probability at least  $1-\frac1{n\log n}$, $Z$ is at most $F^2_{\nsim} \equiv \mu_d + \sqrt{4(n-2)r_s(\log n+\log\log n)}+\frac{2( \log n+\log \log n)}{3}$ when $u,v$ belong to different clusters.
 
We calculate the minimum value of $F^2_\sim(x)$ over all values of $x$ to find the value closest to $F^2_\nsim$. It can be easily observed that $F^2_\sim(x)$ is a decreasing function with the minimum value of $(\frac{n}{2} -2)(1-3r_s) + \frac{n}{2}(1-4r_d)- \sqrt{3nr_s(\log 2n+\log \log n)} - \sqrt{4nr_d(\log 2n+\log \log n)}-\frac{4( \log2n+\log \log n)}{3} $ at $x=r_s$. Plugging in $r_s=\frac{a \log n}{n}$, $r_d=\frac{b \log n}{n}$  we get that the
 algorithm will be successful to resolve correctly  with probability $1- \frac{3}{n\log n}$ as long as,
$$
\left(\frac{n}{2} -2\right)(1-3r_s) + \frac{n}{2}(1-4r_d)- \sqrt{3nr_s(\log 2n+\log \log n)} - \sqrt{4nr_d(\log 2n+\log \log n) }-\frac{4( \log2n+\log \log n)}{3} $$
$$\ge (n-2)(1-2r_s)+ \sqrt{4(n-2)r_s(\log n+\log\log n)}+\frac{2( \log n+\log \log n)}{3}.$$
Plugging in $r_s=\frac{a \log n}{n}$, $r_d=\frac{b \log n}{n}$  and ignoring $o(\log n)$ factors, we get that 
$$a-4b \ge 2(\sqrt{3a}+\sqrt{4b}+\sqrt{4a}+2).$$
\end{proof}

We have done a comparison of the performances of Motif 1, Motifs 2 (or 3) and Motif 4, i.e., the guarantees of Theorems \ref{th:main}, \ref{th:main2} and \ref{th:main3} in Figure \ref{fig:compare}. From the figure it is evident that the Motif 1 outperforms other motifs in terms of current theoretical guarantees.

\begin{figure}
\centering
\includegraphics[scale =0.3]{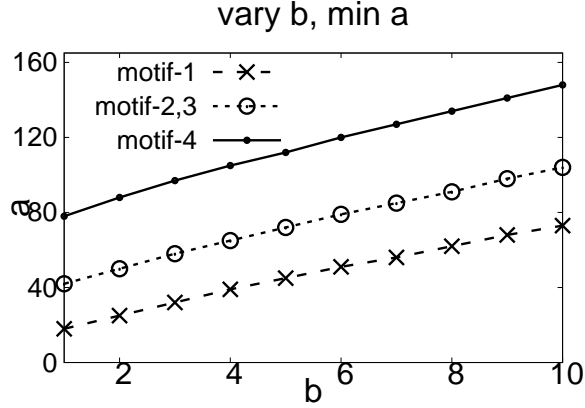}
\caption{A comparison of Theorems \ref{th:main}, \ref{th:main2} and \ref{th:main3}. We vary the value of $b$ and check the minimum value of $a$ that guarantees the condition of success in the respective theorems (motifs).}
\label{fig:compare}
\end{figure}

\end{document}